\documentclass[pra, showkeys, reprint, nofootinbib, superscriptaddress]{revtex4-2}
\usepackage{qcircuit}
\usepackage[english]{babel}
\usepackage[utf8]{inputenc}
\usepackage[colorinlistoftodos, color=green!40, prependcaption]{todonotes}
\usepackage{color}
\usepackage{amsthm}
\usepackage{physics}
\usepackage{subcaption} 
\usepackage{mathtools}
\usepackage{amsmath,amssymb}
\usepackage{xcolor}
\usepackage{graphicx}
\usepackage{soul}
\usepackage{braket}
\usepackage{bbm}
\usepackage{mathbbol}
\usepackage{enumitem}
\usepackage{float}

\usepackage[pdftex, pdftitle={Article}, pdfauthor={Author}, hidelinks]{hyperref}
\usepackage[inkscapelatex=false]{svg}
\DeclareMathOperator{\Span}{span}

\newcommand{\eye}{\mathbb{1}}

\newtheorem{definition}{Definition}
\newtheorem{remark}{Remark}

\newtheorem{lemma}{Lemma}
\newtheorem{proposition}{Proposition}

\begin{document}

\title{Quantum State Certification via Effective Parent Hamiltonians from Local Measurement Data} 

\author{Guy-Philippe Nadon}
\affiliation{École de technologie supérieure, Université du Québec, Montréal, Québec, Canada}

\author{Guanyi Heng}
\affiliation{École de technologie supérieure, Université du Québec, Montréal, Québec, Canada}
\affiliation{Department of Electrical and Computer Engineering, McGill University, Montréal, Québec, Canada}

\author{Pacôme Gasnier}
\affiliation{École de technologie supérieure, Université du Québec, Montréal, Québec, Canada}

\author{Antoine Lemelin}
\affiliation{École de technologie supérieure, Université du Québec, Montréal, Québec, Canada}

\author{Camille Coti}
\affiliation{École de technologie supérieure, Université du Québec, Montréal, Québec, Canada}


\author{Zeljko Zilic}
\affiliation{Department of Electrical and Computer Engineering, McGill University, Montréal, Québec, Canada}
\affiliation{École de technologie supérieure, Université du Québec, Montréal, Québec, Canada}

\author{Mikko Möttönen}
\affiliation{QMill Oy, Keilaranta 12 D, 02150 Espoo, Finland}

\author{Ville Kotovirta}
\affiliation{QMill Oy, Keilaranta 12 D, 02150 Espoo, Finland}

\author{Toni Annala}
\affiliation{Department of Mathematics, University of Toronto, Bahen Centre for Information Technology, 40 St.\ George Street, Room 6290, Toronto, Ontario, Canada M5S 2E4}
\affiliation{QMill Oy, Keilaranta 12 D, 02150 Espoo, Finland}

\author{Ernesto Campos}
\affiliation{École de technologie supérieure, Université du Québec, Montréal, Québec, Canada}
\affiliation{Escuela de Ingenieria y Ciencias, Tecnologico de Monterrey, Av. Heroico Colegio Militar 4700, 31300, Chihuahua, Mexico}

\author{Jacob Biamonte}
\email{jacob.biamonte@etsmtl.ca}
\thanks{Homepage: \url{https://profs.etsmtl.ca/jbiamonte/}}
\affiliation{École de technologie supérieure, Université du Québec, Montréal, Québec, Canada}

\begin{abstract}
The preparation and certification of quantum states is a fundamental challenge across quantum information technology. We introduce a tomography-free state certification method that lower-bounds the fidelity by estimating expectation values of engineered parent-Hamiltonian terms from local measurement data. We apply this framework to construct a parent Hamiltonian that enables certification and variational optimization across the Dicke-state family, which includes the single-excitation $W_n$ state. 
We experimentally validate the framework on IBM quantum hardware, certifying genuine multipartite entanglement for $W_n$ states up to six qubits and establishing positive lower bounds on the state fidelity up to thirteen qubits. For Dicke states with two- and three-excitations, we certify genuine multipartite entanglement up to seven qubits. Within this stringent certification framework, these results constitute among the largest witness-certified demonstrations of such states on a programmable quantum processor.
\end{abstract}

\maketitle

Precise validation and certification of quantum states underpin quantum information technologies. Because full quantum state tomography scales poorly with system size, many alternatives instead bound fidelity or entanglement from restricted measurement data, including entanglement witnesses~\cite{GuhneToth2009,Guhne2010Dicke,Guhne2014Witnessing}, direct fidelity estimation~\cite{James2001,daSilva2011,Guhne_2007}, symmetry-based tomography~\cite{Toth2010PI,Schmied2016}, tensor-network-based tomography \cite{Cramer2010EfficientQST}, reduced-density-matrix methods~\cite{Linden2002,Sawicki2012}, compressed sensing~\cite{Gross2010,Flammia2012}, and randomized measurements~\cite{Huang2020}. These methods replace full reconstruction with operational guarantees while avoiding the exponential cost of complete tomography.

Our approach does so in a particularly direct way. We engineer parent Hamiltonians with a unique target ground state and nonzero spectral gap, then estimate their energies from local observables to obtain certified lower bounds on the target-state fidelity. The Hamiltonian is used only as a classical pseudo-Hamiltonian: it need not be physically implemented, and the protocol requires only local measurements and classical post-processing. Its cost scales polynomially with the number of parent-Hamiltonian terms.

We focus on the $W_n$ state, a standard benchmark with direct experimental relevance due to applications in loss-tolerant quantum networking~\cite{DurVidalCirac2000,Bennett1993Repeaters,Epping2017MultipartiteQKD}. Although $W_n$ and higher Dicke states admit symmetry-based characterization schemes, including permutationally invariant tomography and tailored entanglement witnesses~\cite{Fannes1992,PerezGarcia2007,Guhne2010Fidelity,Moroder2013,Guhne2014Few,Toth2010PI,Schmied2016,GuhneToth2009,Guhne2010Dicke}, our method instead derives fidelity bounds and entanglement certification directly from locally measured parent-Hamiltonian data. It is therefore complementary to inference-based and model-assisted tomography~\cite{toolbox,daSilva2011,lange2023adaptive,lohani2021experimental,cimini2024benchmarking,innan2024quantum}, while deliberately avoiding full state reconstruction, likelihood- or Bayesian-based inference~\cite{hu2025full,catalano2025experimentalpreparationwstates}, and post-processing such as measurement mitigation or zero-noise extrapolation~\cite{nation2021scalable,giurgica2020digital}. The resulting bounds follow directly from measured local data, without optimization over candidate density matrices or classical filtering, rescaling, or inversion of the outcome distribution.  These fidelity bounds then enable entanglement witness evaluation.  

In the next section, we show how measured energies yield fidelity lower bounds, construct parent Hamiltonians and entanglement witnesses for Dicke states, present the experimental results, and conclude with a discussion of the method and its limitations.  Appendix \ref{app:variational_principle} includes the supporting data sets and proofs.  

\section{Stability Bound as a Tomography-Free Certification Tool}

We consider a non-negative Hamiltonian $H$ with a nondegenerate ground state $\ket{\psi}$ and a spectral gap $\Delta \geq 1$.
A stability bound relates the energy expectation value $\Tr(H\rho)$ of a prepared state $\rho$ to its overlap with $\ket{\psi}$.

\begin{lemma}[Stability lower bound]
Let $H$ be a positive semidefinite Hamiltonian with ground-state projector $\ket{\psi}\!\bra{\psi}$ of eigenvalue $0$, and let the spectral gap above the ground state be $\Delta\ge 1$.
Then for any density operator $\rho$ on the same Hilbert space,
\begin{equation}
\label{eq:stability}
\bra{\psi}\rho\ket{\psi}
\;\ge\;
1 - \frac{\Tr(H\rho)}{\Delta}.
\end{equation}
\end{lemma}

In particular, the bound is nontrivial whenever $\Tr(H\rho)\le \Delta$, and it underpinned the universality proofs of feed-forward–restricted variational quantum computation \cite{UVQC}, establishing the ansatz class as a computational model. Here we reinterpret it operationally: Eq.~\eqref{eq:stability} provides a tomography-free fidelity guarantee, computed in the classical outer loop using only local measurement data. 

Crucially, estimating $\Tr(H\rho)$ requires only local measurements since $H$ admits a Pauli decomposition. No state reconstruction, inference procedure, or sampling-based fidelity estimator is required. For more details see Appendix \ref{app:variational_principle}.

\section{Parent Hamiltonians to Certify Dicke States}

Given a desired target state $\ket{\psi}$, any Hamiltonian $H$ satisfying the assumptions of the stability lemma can be used to certify any state $\ket{\psi}$.

Operationally, the Hamiltonian is measured term-wise via local measurements,
\begin{equation}
H = \sum_l h_l \sigma_l,
\end{equation}
where $\sigma_l$ are Pauli strings.
The expectation value $\Tr(H\rho)$ is obtained by classical post-processing of measurement outcomes.  This method appears regularly in variational quantum computation \cite{cerezo2021variational, UVQC, 2017Natur.549..242K}.

The lower bound in Eq.~\eqref{eq:stability} then serves two roles:
(i) it provides a certified lower bound on the fidelity,
(ii) it functions as an entanglement witness, since sufficiently low energy excludes all separable states (see Appendix~\ref{app:witness}).

The Dicke states are stringent benchmarks whose parent Hamiltonians have been studied extensively within tensor-network formalisms~\cite{Verstraete2006PEPS}.
Originating in condensed-matter and quantum-optical physics~\cite{Dicke1954Superradiance}, they are highly entangled outside the $0$- and $n$-excitation sectors and non-stabilizer~\cite{Toth2007DickeWitness} with applications in loss-tolerant quantum networking~\cite{DurVidalCirac2000,Bennett1993Repeaters,Epping2017MultipartiteQKD}. 

The $n$-qubit Dicke state with $k$ excitations is
\begin{equation}
\label{eq:Dicke-state}
\begin{split}
\ket{D_n^{(k)}}
&=
\binom{n}{k}^{-1/2}
\sum_{\substack{x \in \{0,1\}^n \\ |x| = k}}
\ket{x}
\\
&=
\binom{n}{k}^{1/2}
\, \Pi_{\mathrm{sym}}
\, \ket{1^k}\ket{0^{\,n-k}}
\end{split}
\end{equation}
where the sum runs over all computational basis states of Hamming weight $k$, and
\(
\Pi_{\mathrm{sym}} = \tfrac{1}{n!}\sum_{\chi \in S_n} P_\chi
\)
is the projector onto the fully symmetric subspace $S_n$.

We use a parent Hamiltonian, the ground-state space of which is $\ket{D_n^{(k)}}$,
\begin{equation}
\label{eq:Dicke_hamiltonian}
H_n^{(k)} = \frac{1}{n} \sum_{j<\ell} (\eye - S_{j\ell}) + (P - k \cdot \eye)^2,
\end{equation}
where $S_{j\ell}$ is the swap operator of the state of qubits $j$ and $\ell$, and $P$ is the Hamming-weight operator,
\begin{equation}
P = \sum_{j=1}^n \ket{1}_j\!\bra{1},
\end{equation}
where the notation $\ket{1}_j\!\bra{1}$ means the projector $\ket{1}\!\bra{1}$ acting on the $j$th qubit and identity everywhere else.

This Hamiltonian is related to the Lipkin--Meshkov--Glick model~\cite{Moroder2013Certification, Ribeiro2008LMG}, differing by the scale factor $1/(2n)$ to force a constant unit gap in the first term and the choice of excitation sector, parameterized by $k$ in the second term.
Through the introduction of the tuning parameter $k$, one selects which Dicke state spans the kernel of the Hamiltonian, thereby enabling certification of the full family of Dicke states.  A proof of the following is given in Appendix~\ref{app:W_Hamiltonian}: 

\begin{proposition}
The Hamiltonian~\eqref{eq:Dicke_hamiltonian} is a frustration-free parent Hamiltonian for the $n$-qubit Dicke state with $k$ excitations and satisfies the conditions of Proposition~\ref{prop:variational}.
\end{proposition}

\section{Dicke Family Entanglement Witnesses}

We certify genuine multipartite entanglement using a standard fidelity-based witness. 
An entanglement witness is a Hermitian operator whose expectation value is nonnegative on all biseparable states, but negative for at least one genuinely $n$-partite entangled state.

For a target normalized pure state $\ket{\psi}$ and an experimentally prepared state $\rho$, the fidelity is
\begin{equation}
F_\psi = \bra{\psi}\rho\ket{\psi}.
\end{equation}

A generic witness detecting $\ket{\psi}$ is given by \cite{En_det}
\begin{equation}\label{eqn:witness}
\mathcal{W} = \alpha\,\eye - \ket{\psi}\!\bra{\psi},
\end{equation}
where $\alpha$ is the maximal fidelity achievable by any biseparable state,
\begin{equation}
\alpha = \max_{\ket{\phi}\in\mathcal{B}} |\braket{\phi|\psi}|^2.
\end{equation}
Here, $\mathcal{B}$ denotes the set of pure product states across all bipartitions $A\!:\!B$ of the $n$-qubit system.

By construction, $\Tr(\mathcal{W}\rho)\ge 0$ for all biseparable states. 
Thus, observing $\Tr(\mathcal{W}\rho)<0$ certifies genuine $n$-partite entanglement \cite{alpha_calculation}.

For a fixed bipartition $A\!:\!B$, the maximal fidelity with product states is given by the largest Schmidt coefficient of $\ket{\psi}$ across that cut 
and the witness parameter is obtained by maximizing over all bipartitions \cite{toolbox}. 

\begin{proposition}
\label{alphadickeprop}
For the Dicke state $\ket{D_n^{k}}$ with $n$ qubits and $k$ excitations, the maximal biseparable fidelity is
\[
\alpha =
\begin{cases}
\displaystyle \frac{n-k}{n}, & k<\frac{n}{2},\\[6pt]
\displaystyle \frac{n}{2(n-1)}, & k=\frac{n}{2}\ \text{(even $n$)},\\[8pt]
\displaystyle \frac{k}{n}, & k>\frac{n}{2}.
\end{cases}
\]
\end{proposition}

The derivation follows from the Schmidt decomposition of Dicke states across arbitrary bipartitions and is confirmed in Appendix~\ref{app:witness}.

\section{Experimental Results}

We experimentally validate tomography-free certification using effective parent Hamiltonians
evaluated from local measurement data.
As a benchmark family we consider the $n$-qubit $W$ states $\ket{W_n}=\ket{D_n^{(1)}}$
for system sizes $n=2, 3,\dots,16$.
For each $n$, we (i) prepare a state $\rho_n$ using a compiled W-state circuit on superconducting
hardware, (ii) estimate the energy $\langle H_n\rangle=\Tr(H_n\rho_n)$ of an effective parent
Hamiltonian $H_n$ using only local Pauli measurements and classical post-processing, and
(iii) convert the measured energy into a certified lower bound on the fidelity with $\ket{W_n}$
using the stability bound.

All experiments were executed on IBM \texttt{ibm\_quebec} (Heron r2) device.
For each $n$ we choose a connected set of $n$ physical qubits on the device coupling graph
(Fig.~\ref{fig:qpu_layout}) to avoid SWAP overhead.
All circuits are transpiled into the native gate set using a fixed layout and routing strategy
to ensure comparable depth across $n$.

We prepare $\ket{W_n}$ using circuits compiled from a standard W-state construction
into the native gate set of the target processor.
A representative example for $n=4$ is shown in Fig.~\ref{fig:W_circuits}.
For each $n$, the compiled circuit contains $2n-3$ two-qubit entangling gates and $7n-8$
single-qubit gates.

For certification we evaluate an effective parent Hamiltonian $H_n$ whose unique ground state
is $\ket{W_n}$ and whose spectral gap satisfies $\Delta=1$.
We use the $k=1$ member of the Dicke-family parent Hamiltonians, written in Pauli form as
\begin{align}
H_n
&= C \cdot \eye 
-\Big(\frac{n}{2}-1\Big)\sum_{j=1}^n Z_j
-\frac{1}{2n}\sum_{j<\ell}\big(X_jX_\ell + Y_jY_\ell\big)
\notag\\
&\quad
+\Big(\frac{1}{2}-\frac{1}{2n}\Big)\sum_{j<\ell} Z_jZ_\ell,
\label{eq:Hn_pauli}
\end{align}
with $C=\big(\frac{n}{2}-1\big)^2+\frac{2n-1}{4}$.
Accordingly, the measured energy is obtained from one- and two-body correlators:
\begin{align}
\langle H_n\rangle
&= C
-\Big(\frac{n}{2}-1\Big)\sum_{j=1}^n \langle Z_j\rangle
-\frac{1}{2n}\sum_{j<\ell}\big(\langle X_jX_\ell\rangle + \langle Y_jY_\ell\rangle\big)
\notag\\
&\quad
+\Big(\frac{1}{2}-\frac{1}{2n}\Big)\sum_{j<\ell} \langle Z_jZ_\ell\rangle.
\label{eq:energy_from_correlators}
\end{align}

All terms in \eqref{eq:Hn_pauli} are products of identical Pauli operators.
Therefore three global measurement settings suffice for all $n$:
measuring all qubits in the $Z$ basis yields $\{\langle Z_j\rangle,\langle Z_jZ_\ell\rangle\}$, measuring all qubits in the $X$ basis yields $\{\langle X_jX_\ell\rangle\}$, and measuring all qubits in the $Y$ basis yields $\{\langle Y_jY_\ell\rangle\}$. Each basis was executed with $S=16384$ shots, for a total of $3S$ shots per system size.

We construct $H_n$ as a sum of sparse Pauli operators and evaluate $\langle H_n\rangle$ using  Qiskits' runtime primitives. In our case, the Pauli terms partition into three commuting groups associated with the $X$, $Y$, and $Z$ measurement bases, so the estimator can evaluate the energy from three basis settings and return $\langle H_n\rangle$ computed from shot averages.

By construction, $H_n\succeq 0$ has unique ground state $\ket{W_n}$ and a constant spectral gap
$\Delta=1$. The stability bound therefore yields a certified fidelity lower bound from the measured energy:
\begin{align}
& F_n := \bra{W_n}\rho_n\ket{W_n}
\;\ge\;
F^{\mathrm{lb}}_n
\notag\\
&\quad := \max\Big\{0,\;1-\frac{\langle H_n\rangle}{\Delta}\Big\}
= \max\{0,\;1-\langle H_n\rangle\}.
\label{eq:fidelity_lb}
\end{align}
We certify genuine multipartite entanglement using the fidelity witness threshold for $W_n$ states,
\begin{equation}
\alpha_n = \max_{\text{biseparable }\sigma}\bra{W_n}\sigma\ket{W_n} = 1-\frac{1}{n},
\end{equation}
so that $F^{\mathrm{lb}}_n>\alpha_n$ certifies genuine $n$-partite entanglement.

Statistical uncertainty arises from finite-shot estimation of correlators in
\eqref{eq:energy_from_correlators}.
Error bars in Fig.~\ref{fig:W_experiment} represent one standard error of the mean (SEM)
computed from shot noise.

Concretely, for each basis $b\in\{X,Y,Z\}$ and each shot $s$, the measurement yields outcomes
$m^{(s)}_j\in\{\pm1\}$ for all qubits. We compute the relevant per-shot contribution
$g_b^{(s)}$ to \eqref{eq:energy_from_correlators} (e.g., sums of $m^{(s)}_jm^{(s)}_\ell$
with the appropriate coefficients), and estimate the basis contribution by the sample mean
$\widehat{g}_b = \frac{1}{S}\sum_{s=1}^S g_b^{(s)}$.
The SEM for each basis is $\mathrm{SE}(\widehat{g}_b)=\sqrt{\widehat{\mathrm{Var}}(g_b)/S}$,
and the total energy SEM is obtained by combining independent bases in quadrature:
\begin{equation}
\mathrm{SE}(\langle H_n\rangle)
=
\sqrt{\mathrm{SE}(\widehat{g}_X)^2+\mathrm{SE}(\widehat{g}_Y)^2+\mathrm{SE}(\widehat{g}_Z)^2}.
\end{equation}
Since $F^{\mathrm{lb}}_n=\max\{0,1-\langle H_n\rangle\}$ for $\Delta=1$, the uncertainty on
$F^{\mathrm{lb}}_n$ is inherited directly from that of $\langle H_n\rangle$ in the nontrivial regime
$\langle H_n\rangle<1$.

Figure~\ref{fig:W_experiment} reports $\langle H_n\rangle$ and the derived bounds
$F^{\mathrm{lb}}_n$ for $n=2,\dots,16$.
Certification is nontrivial whenever $\langle H_n\rangle<\Delta=1$, in which case
$F^{\mathrm{lb}}_n=1-\langle H_n\rangle$ provides a guaranteed fidelity lower bound without
tomographic reconstruction.
Entanglement certification holds when $F^{\mathrm{lb}}_n>\alpha_n$, and breaks down at larger $n$
as two-qubit gate noise accumulates with circuit depth.

\section{Discussion}

We introduced a tomography-free framework for quantum state certification based on parent Hamiltonians evaluated from local measurement data. From the same measured observables, the method yields rigorous lower bounds on target-state fidelity together with entanglement-witness tests, eliminating the need for full quantum state tomography. Because the guarantees follow directly from the spectral properties of the Hamiltonian, the approach is general and applies to arbitrary Hamiltonians with a gapped nondegenerate ground state, independent of the specific physical implementation.

Experimentally, we applied this framework to Dicke states prepared on a programmable quantum processor. Using only local correlators, we certify genuine multipartite entanglement for $W_n$ states up to six qubits and obtain nontrivial fidelity lower bounds for systems as large as thirteen qubits. For higher Dicke states with two and three excitations, we certify multipartite entanglement up to seven qubits while maintaining meaningful fidelity bounds for larger system sizes. Within this stringent certification setting, these results represent some of the largest witness-certified realizations of such states on programmable quantum hardware.

Because the guarantees rely only on experimentally accessible local observables, the framework integrates naturally with existing characterization techniques and complements reconstruction- and mitigation-based approaches designed for accurate observable estimation. More broadly, parent-Hamiltonian certification provides a scalable route to validating structured many-body quantum states without full tomography, offering a practical tool for verifying increasingly complex quantum states on near-term quantum devices. In this way, rigorous certification of large entangled states can be achieved directly from experimentally accessible local measurements.

\onecolumngrid

\begin{figure}[htbp]
  \centering

  \includegraphics[width=0.4\columnwidth]{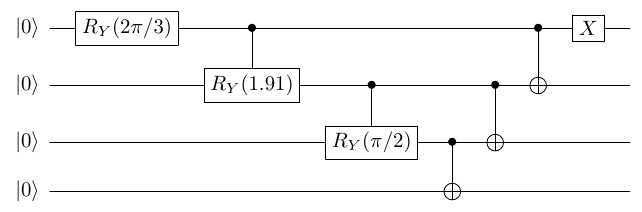}

  \vspace{0.8em}
  {\small{Textbook circuit for a 4-qubit $W_4$ state}}

  \vspace{1.2em}

  \includegraphics[width=\columnwidth]{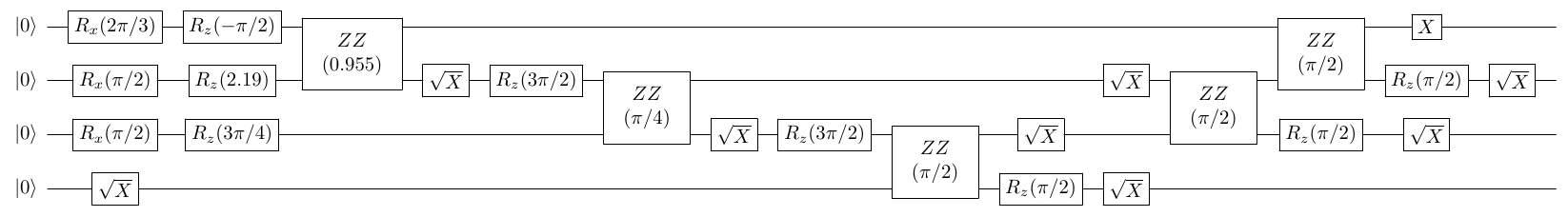}

  \vspace{0.8em}
  {\small{Native gate-set circuit for a 4-qubit $W_4$ state}}

  \caption{State-preparation circuits for $W_4$ states.
  The lower circuit shows the compilation of the textbook construction into the
  native gate set of the superconducting processor.}
  \label{fig:W_circuits}
\end{figure}

\begin{figure}[ht!] 
\centering 
\includegraphics[width=\linewidth]{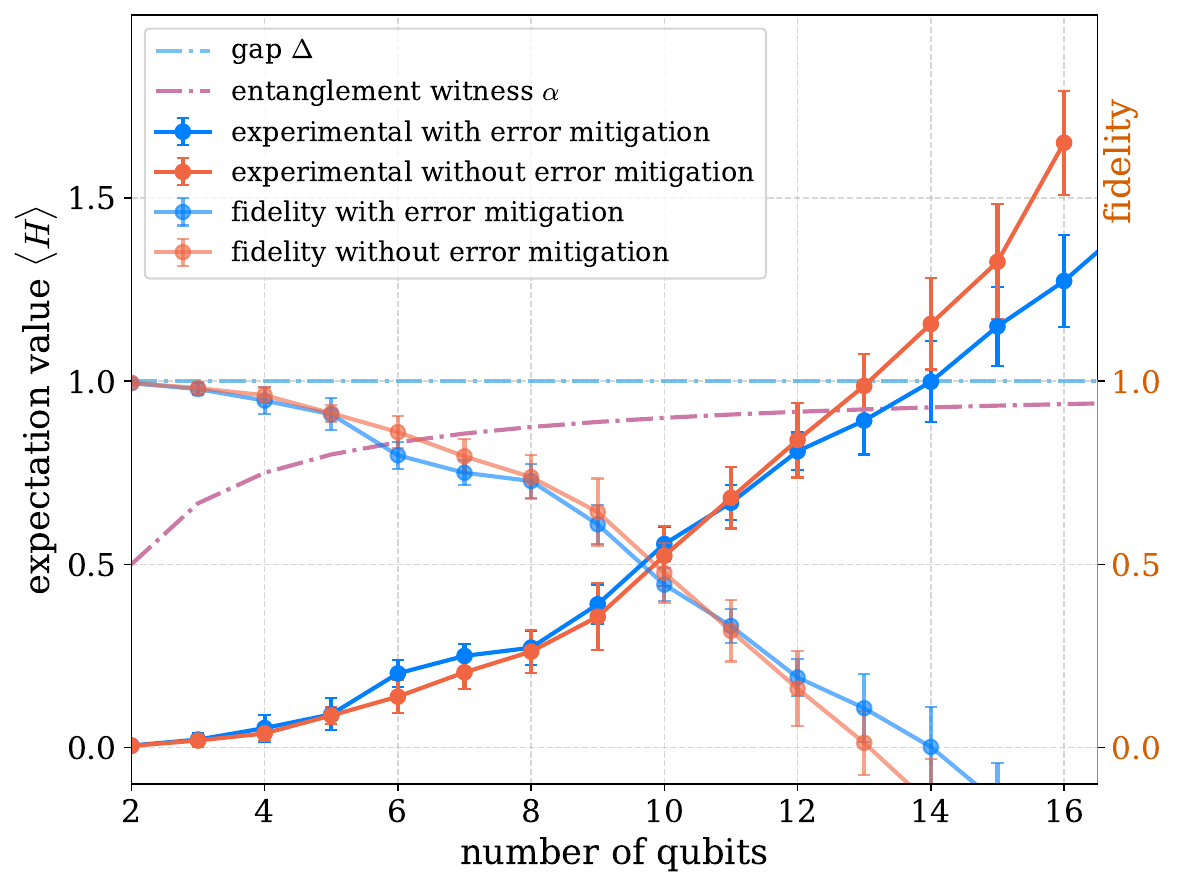}  \caption{Expectation value versus number of qubits for prepared $W_n$ states. The blue dashed line indicates the spectral gap $\Delta=1$, and the purple curve shows the entanglement-witness threshold $\alpha=1-1/n$. } 

\label{fig:W_experiment} 
\end{figure}

\section{Acknowledgment}

The authors thank Sean Wagner, Ibrahim Shehzad, and the IBM team for valuable feedback. This research was enabled in part by support from Calcul Québec and the Digital Research Alliance of Canada. We also thank the Plateforme d’Innovation Numérique et Quantique du Québec (PINQ²) for access to the ibm\_quebec system and the computational resources required for this study. This work was partially supported by Mitacs/PINQ² (Project IT42780) and by the MEIE Principal Ministry Research Chair in Quantum Computing held by J.B.~at ÉTS Montréal, Université du Québec.  The code used to generate the numerical results in this study, along with the associated datasets, is available on GitHub~\cite{QuantumStateCertParentHamiltonian2026}.  

\newpage 

\bibliographystyle{unsrt}
\bibliography{refs}

\newpage 

\appendix

\section{Variational principle}
\label{app:variational_principle}
\begin{proposition}\label{prop:variational}
    Let $H \in \text{herm}_{\mathbb C} (2^n \times 2^n)$.
    Let normalized $\ket{\psi} \in [\mathbb{C}^2]^{\otimes n}$, and let $\rho$ be a normalized density state in $\text{herm}_{\mathbb C} (2^n \times 2^n)$.
    We have:
    \begin{enumerate}
        \item $H$ has a 1D kernel, $\ker{H} = \Span\{{\ket{\psi}}\}$,
        \item $H$ is non-negative, that is $H \succeq 0$,
        \item $H$ is gapped, that is all non-zero eigenspaces of $H$ are at least $\Delta \geq 1$, that is $(\ker{H})^\perp \geq \Delta$.
    \end{enumerate}
    If $\rho$ is such that 
    \begin{equation}
        \Tr\{H \rho\} < \Delta, 
    \end{equation}
    then
    \begin{equation}
         1 - \frac{\Tr\{\rho H\}}{\Delta} \leq \bra{\psi}\rho\ket{\psi}\leq 1- \frac{\Tr\{\rho H\}}{\max H},\label{eqn:bounds}
    \end{equation}
    where $\max H$ is the largest eigenvalue of $H$.  
\end{proposition}

A closely related variational theorem was developed by one of the authors to show that a model of computation based on expectation-value measurements is computationally universal \cite{UVQC}. The present theorem originates in that setting, while the lower bound \eqref{eqn:bounds} traces back to earlier work of Eckart \cite{Eckart1930}. Here we include an elementary proof of the bound. 

\begin{proof}
    Express $H$ using the spectral decomposition where the eigenvalues are ordered monotonically increasing in $k$, 
    \begin{align}
\qquad
H &= \sum_{k} \lambda_k \, |\psi_k\rangle\langle\psi_k|, 
\quad \lambda_0 < \lambda_1 \le \lambda_2 \le \cdots,\quad \lambda_1 = \Delta \geq 1,\quad  \lambda_0 = 0 \label{eq:spectral},
\qquad
\end{align}
and express the density matrix $\rho$ in the eigenbasis of H, 
\begin{align}
\rho &\mapsto \sum_{k} \rho_k \, |\psi_k\rangle\langle\psi_k|, \quad
\rho_k = \langle\psi_k|\rho|\psi_k\rangle,\;
\rho_k \ge 0,\;
\sum_k \rho_k = 1 \label{eq:rho-decomp}.
\end{align}
From \eqref{eq:spectral} and \eqref{eq:rho-decomp}
\begin{equation}
 \Tr{\rho H} = \sum_{k\geq0} \lambda_k \rho_k = \sum_{k>0} \lambda_k \rho_k.\label{eq:energy}
\end{equation}
\text{From \eqref{eq:rho-decomp}, \eqref{eq:energy} and $\frac{\lambda_k}{\Delta}\geq1$ for $k>0$, we establish the following inequality} 
\begin{equation}
    \sum_{k>0}\rho_k \leq \frac{1}{\Delta}\sum_{k>0} \lambda_k \rho_k.\label{eq:ineq_lambda_rho}
\end{equation}
We note that
\begin{equation}
    \Tr{\rho}=\rho_0 + \sum_{k>0}^{2^n-1} \rho_k = 1,  \label{eq:trace_rho}
\end{equation}
then construct the following inequality obtained by using \eqref{eq:ineq_lambda_rho} and \eqref{eq:trace_rho} 
\newline where $\frac{\lambda_k}{\Delta}\geq1, \forall  k>0$ 
\begin{equation}
1 \leq \rho_0 + \frac{\sum_{k>0} \lambda_k \rho_k}{\Delta} \label{eq:inequality}
\end{equation}
\text{and hence}
\begin{equation}
1\leq\rho_0 + \frac{\Tr{\rho H}}{\Delta}. 
\label{eq:bound-general}
\end{equation}
\end{proof}

\section{Entanglement witnesses}
\label{app:witness}

\begin{proposition}
The maximal fidelity with the product states across any partition is
\begin{equation}
    F_{\max}(A:B) \;=\; \Big(\max_k \sqrt{\lambda_k}\Big)^2
    \;=\; \max_k \lambda_k,
\end{equation}
and the witness parameter is obtained by maximizing over all bipartitions \cite{toolbox}:
\begin{equation}
    \alpha \;=\; \max_{A:B}\;(\max_k \lambda_k).
\end{equation} 
\end{proposition}

\begin{proof}
Let $\ket{\psi}$ be a pure state and fix a bipartition $A:B$, and let $\ket{\phi}=\ket{a}_A\otimes\ket{b}_B$ be an arbitrary normalized product state across this cut. 
Choosing an orthonormal product basis $\{\ket{\ell}_A \otimes \ket{m}_B\}_{\ell,m}$ for this cut, 
we may write
\begin{equation}\label{eq:state_expansion}
|\psi\rangle \;=\; \sum_{\ell,m} c_{\ell m}\,|\ell\rangle_A |m\rangle_B,
\qquad
|\phi\rangle \;=\; \Big(\sum_{\ell} a_{\ell}\,|\ell\rangle_A\Big)\otimes\Big(\sum_{m} b_{m}\,|m\rangle_B\Big).
\end{equation}

where $C=(c_{\ell m})$ is the coefficient matrix and $\mathbf{a}=(a_\ell)$, $\mathbf{b}=(b_m)$ are normalized vectors with $\|\mathbf{a}\|_2=\|\mathbf{b}\|_2=1$. 
The overlap is
\begin{equation}
    \braket{\phi|\psi}
    \;=\;\sum_{\ell m} a_\ell^{*}\,c_{\ell m}\,b_m^{*}
    \;=\;\bra{a}\,C\,\ket{b^{*}}.
\end{equation}
Maximizing over normalized $\mathbf{a},\mathbf{b}$ gives
\begin{equation}
    \max_{\ket{\phi}}
    \big|\braket{\phi|\psi}\big|
    \;=\;\max_{\substack{\|\mathbf{a}\|_2=1\\ \|\mathbf{b}\|_2=1}}
    \big|\bra{a}\,C\,\ket{b^{*}}\big|
    \;=\;\max_k \sqrt{\lambda_k(CC^\dagger)},
\end{equation}
where $\{\lambda_k(CC^\dagger)\}$ are the eigenvalues of $CC^\dagger$. 
Thus, the Schmidt coefficients of $\ket{\psi}$ with respect to the bipartition $A:B$ are precisely $\{\sqrt{\lambda_k(CC^\dagger)}\}$ \cite{alpha_calculation}.
\end{proof}

\begin{proposition}\label{prop:Dicke_max_biseparable_fidelity}The maximal biseparable fidelity $\alpha$ for the Dicke state $\ket{D_n^{k}}$, where $n$ denotes the number of qubits and $k$ denotes the number of excitations, is

\[
\alpha \;=\;
\begin{cases}
\displaystyle \frac{n-k}{n}, & k<\frac{n}{2},\\[6pt]
\displaystyle \frac{n}{2(n-1)}, & k=\frac{n}{2}\ \text{(even $n$)},\\[8pt]
\displaystyle \frac{k}{n}, & k>\frac{n}{2}.
\end{cases}
\]
\end{proposition}

\begin{proof}
(1) Bipartition (Schmidt) decomposition and what must be maximized.
Fix a bipartition $A:B$ with $|A|=a$, $|B|=b$ and $a+b=n$.
Grouping computational basis states by the number $\beta$ of excitations in subsystem $A$
gives the Schmidt decomposition
\begin{equation}
\label{eq:Dicke-bipartition}
\ket{D_n^{k}} \;=\;
\sum_{\beta}
\sqrt{\lambda_{a:b}(\beta)}\;
\ket{D_{a}^{\beta}}\otimes \ket{D_{b}^{\,k-\beta}},
\qquad
\lambda_{a:b}(\beta)
=
\frac{\binom{a}{\beta}\binom{b}{k-\beta}}{\binom{n}{k}},
\end{equation}
where the sum runs over all $\beta$ such that
$0\le \beta \le a$ and $0\le k-\beta \le b$.
For any pure bipartite state, the maximal overlap with product states across the cut
equals the largest squared Schmidt coefficient. Hence the biseparable fidelity constant is
\begin{equation}
\label{eq:alpha-max}
\alpha
=
\max_{A:B}\;F_{\max}(A:B)
=
\max_{\substack{a, \beta}}
\lambda_{a:\,n-a}(\beta)
=
\max_{\substack{a, \beta}}
\frac{\binom{a}{\beta}\binom{n-a}{k-\beta}}{\binom{n}{k}}.
\end{equation}

(2) Maximization over cuts (three regimes).

Case 1: $k<\frac{n}{2}$.
From ~\cite{BergmannGuehne2013_DickeCriteria},
the maximization in \eqref{eq:alpha-max} is attained at the bipartition that has $a=1$ and $b=(n-1)$.
Evaluating the corresponding term in \eqref{eq:alpha-max} gives the estimate 
\begin{equation}
\label{eq:case1-bound}
\frac{\binom{1}{\beta}\binom{n-1}{k-\beta}}{\binom{n}{k}}
\;\le\;
\frac{n-k}{n},
\qquad \Bigl(k<\frac{n}{2}\Bigr),
\end{equation}
which is tight for the maximizing choice of $\beta$~\cite{BergmannGuehne2013_DickeCriteria}.
Therefore,
\[
\alpha=\frac{n-k}{n},
\qquad \Bigl(k<\frac{n}{2}\Bigr).
\]

Case 2: $k=\frac{N}{2}$ (balanced Dicke state, even $N$).
For the balanced case,~\cite{Toth2007_JOSAB} shows that the maximum in
\eqref{eq:alpha-max} is achieved for bipartition that has $a=2$ and $b=(N-2)$.
In particular, the proof yields the bound
\begin{equation}
\label{eq:case2-bound}
\frac{\binom{2}{\beta}\binom{N-2}{k-\beta}}{\binom{N}{k}}
\;\le\;
\frac{N}{2(N-1)},
\qquad \Bigl(k=\frac{N}{2}\Bigr),
\end{equation}
which is tight at the maximizing value of $\beta$ identified there.
Hence,
\[
\alpha=\frac{N}{2(N-1)},
\qquad \Bigl(k=\frac{N}{2},\; N\ \text{even}\Bigr).
\]

Case 3: $k> N/2$.
Use the local-unitary symmetry
\(
\ket{D_N^{k}} = X^{\otimes N}\ket{D_N^{N-k}}.
\)
Local unitaries do not change Schmidt coefficients (hence do not change $\alpha$), so
$\alpha(N,k)=\alpha(N,N-k)$.
Since $N-k< N/2$ in this regime, Case~1 gives
\[
\alpha(N,k)=\alpha(N,N-k)=\frac{N-(N-k)}{N}=\frac{k}{N}.
\]

Combining the three cases proves the stated formula.
\end{proof}

\section{Dicke states parent Hamiltonians} 
\label{app:W_Hamiltonian}

\begin{definition}[Hamming weight operator] The Hamiltonian 
\begin{equation}\label{eqn:projh}
P = \sum_{j=1}^n \ket{1}_j\!\bra{1} = \frac{1}{2}\left(n\cdot {\eye} - \sum_{j=1}^n   Z_j \right) 
\end{equation}
satisfies $P \ket{x} = \|x\|_1 \ket{x}$ for $n$-long bit strings $x$ where $\|\cdot \|_1$ is the $1$-norm.
\end{definition}

\begin{remark}
The operator $P$ has all the desired properties as listed in Proposition \ref{prop:variational}. The number of terms when $P$ is expressed in the Pauli basis is $n+1$.  
\end{remark}

\begin{proposition} The Hamiltonian \eqref{eq:HW-full}: 
\begin{equation}
\label{eq:HW-full}
H^{(k)}_n = C \mathbb{1} - \left( \frac{n}{2} - k \right) \sum_{j=1}^{n} Z_j 
- \frac{1}{2n} \sum_{j<k} (X_j X_k + Y_j Y_k) 
+ \left( \frac{1}{2} - \frac{1}{2n} \right) \sum_{j<k} Z_j Z_k,
\end{equation}
where
\begin{equation*}
C = \left( \frac{n}{2} - k \right)^2 + \frac{2n - 1}{4},
\end{equation*}
is a frustration-free parent Hamiltonian for the $n$-qubit Dicke states with $k$ excitations which satisfies the conditions of Proposition \ref{prop:variational}.

\end{proposition}

\begin{proof}
We write the Hamiltonian \eqref{eq:HW-full} as $H_W = H_1 + H_2$, where
\begin{equation}
\label{eq:H1-H2}
H_1 = \frac{1}{n} \sum_{j< k} (\eye - S_{jk})
\qquad
H_2 = (P - k\cdot \eye)^2,
\end{equation}
with $\mathrm{S}_{jk}$ the swap operator and $P$ is the Hamming weight operator \eqref{eqn:projh}.  $H_1$ is a sum of non-negative projectors, and $H_2$ is quadratic, so $H_W$ is necessarily non-negative.

For any symmetric state $|\phi\rangle$, $\mathrm{S}_{jk}|\phi\rangle = |\phi\rangle$. The sum $\sum_{j<k}\mathrm{S}_{jk}$ attains
its maximum eigenvalue $\binom{n}{2}$ on the symmetric
subspace, which gives
\begin{equation}
\label{eq:kerH1}
\ker H_1
=
\mathrm{Span}\{\,|D_m^{(n)}\rangle : m = 0,1,\ldots,n\,\},
\end{equation}
where $|D_m^{(n)}\rangle$ are the Dicke states and $m$ denotes the number of excitations.
For any computational basis state $|x\rangle$,
\begin{equation}
\label{H2EIGEN}
H_2|x\rangle = (\|x\|_1 - 1)^2 |x\rangle,
\end{equation}
from which it follows that
\begin{equation}
\label{eq:kerH2}
\ker H_2
=
\mathrm{Span}\{\,|x\rangle : \|x\|_1 = 1\,\}.
\end{equation}
Since $[H_1, H_2]=0$ we can consider the intersection of the two kernels, yielding: 
\begin{equation}
\label{eq:kerHW}
\ker H_W
=
\ker H_1 \cap \ker H_2
=
\mathrm{Span}\{|W_n\rangle\},
\end{equation}
where $\ket{W_n}=\ket{D_1^{(n)}}$.
\paragraph*{Spectral gap.}
From Eq.~\eqref{H2EIGEN}, the eigenvalues of $H_2$ are given by
$(\|x\|_1 - 1)^2$, and therefore the spectral gap of $H_2$ is equal to $1$.
The Hamiltonian $H_1$ is non-negative and its spectral gap is well known to be $1$. 
Consequently, any state orthogonal to the ground space of $H_W = H_1 + H_2$
must have energy at least $1$, implying that
$\Delta \ge 1$.
\end{proof} 

\section{Dicke states result} 
\label{app:Dicke_Result}

\begin{figure}[H]
    \centering
    \begin{subfigure}[t]{0.48\textwidth}
        \centering
        \includegraphics[width=\linewidth]{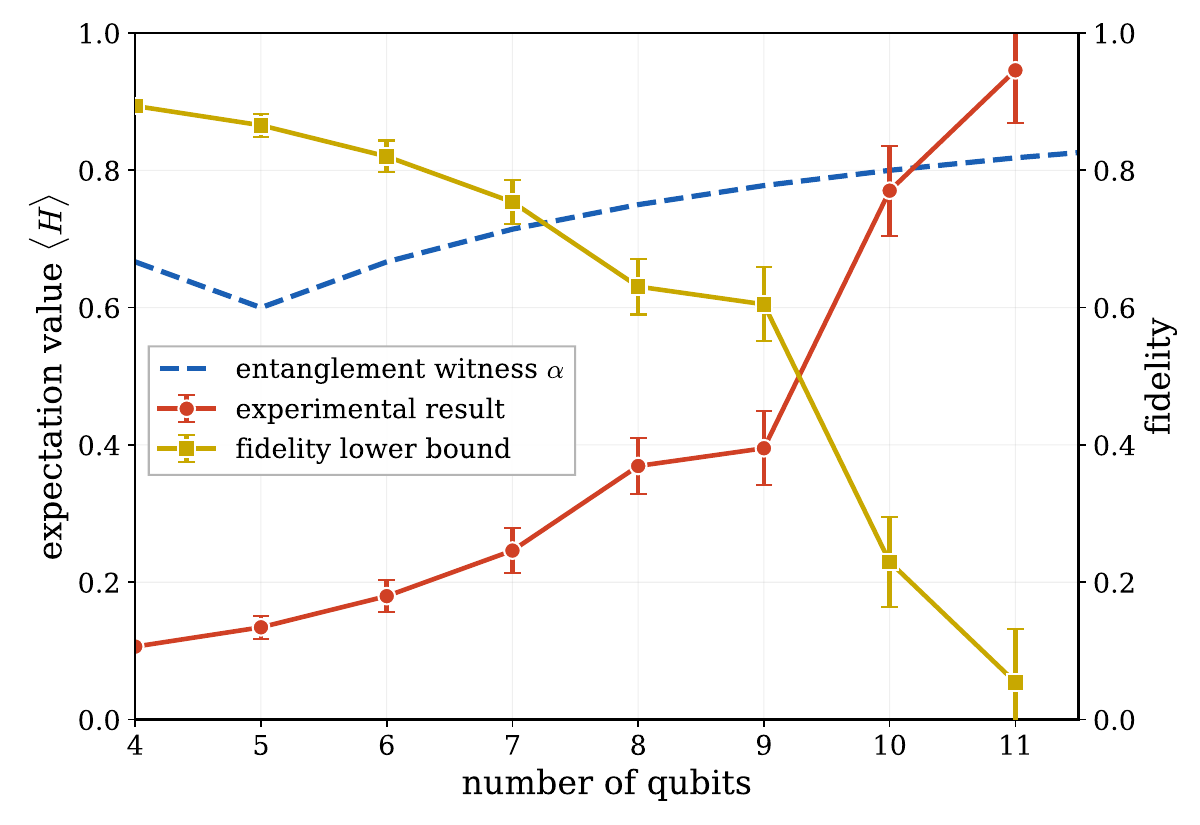}
        \caption{Dicke state with excitation number $k=2$.}
        \label{fig:dicke_k2}
    \end{subfigure}
    \hfill
    \begin{subfigure}[t]{0.48\textwidth}
        \centering
        \includegraphics[width=\linewidth]{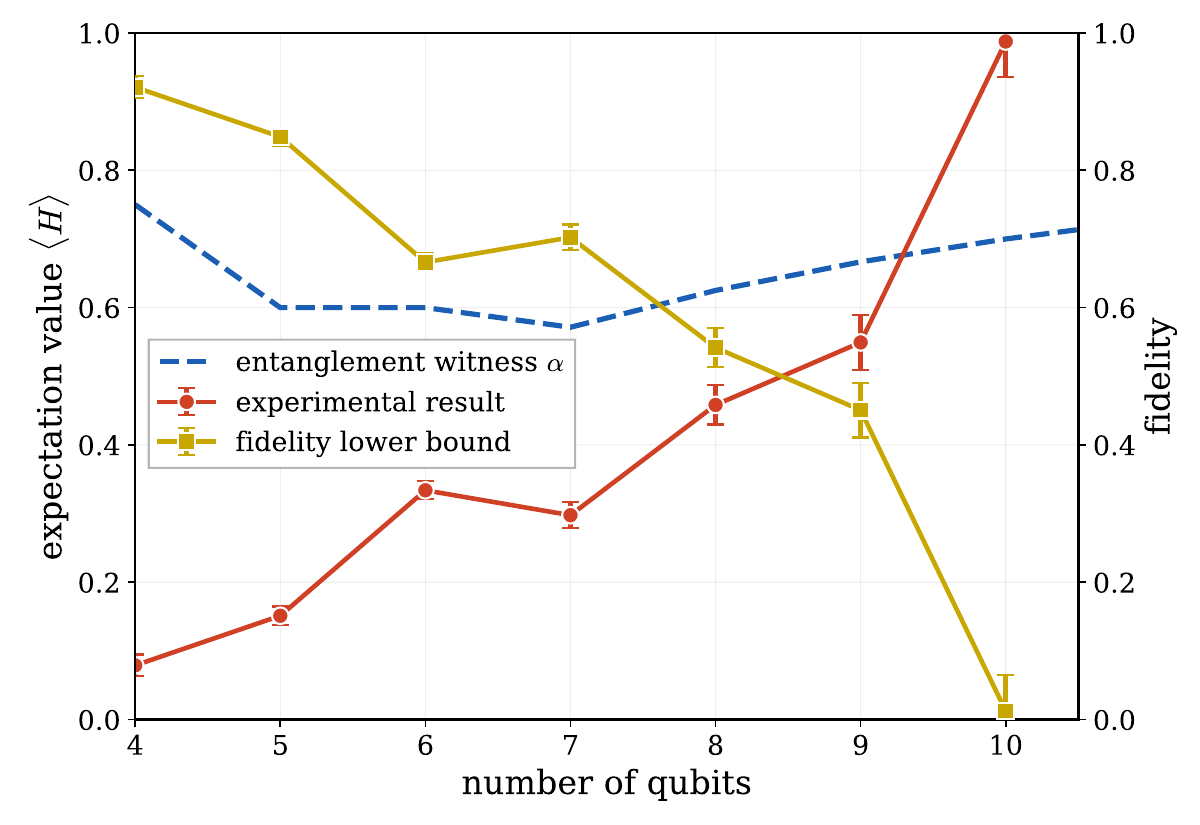}
        \caption{Dicke state with excitation number $k=3$.}
        \label{fig:dicke_k3}
    \end{subfigure}
    
    \caption{
    Expectation value versus the number of qubits for prepared Dicke states.
    Panel (a) corresponds to excitation number $k=2$, and panel (b) to $k=3$.
    The blue dashed curve indicates entanglement-witness threshold as introduced in Proposition~\ref{prop:Dicke_max_biseparable_fidelity}.
    }
    \label{fig:dicke_expectation}
\end{figure}
\begin{figure}[ht!] 
\centering 
\includegraphics[width=1\linewidth]{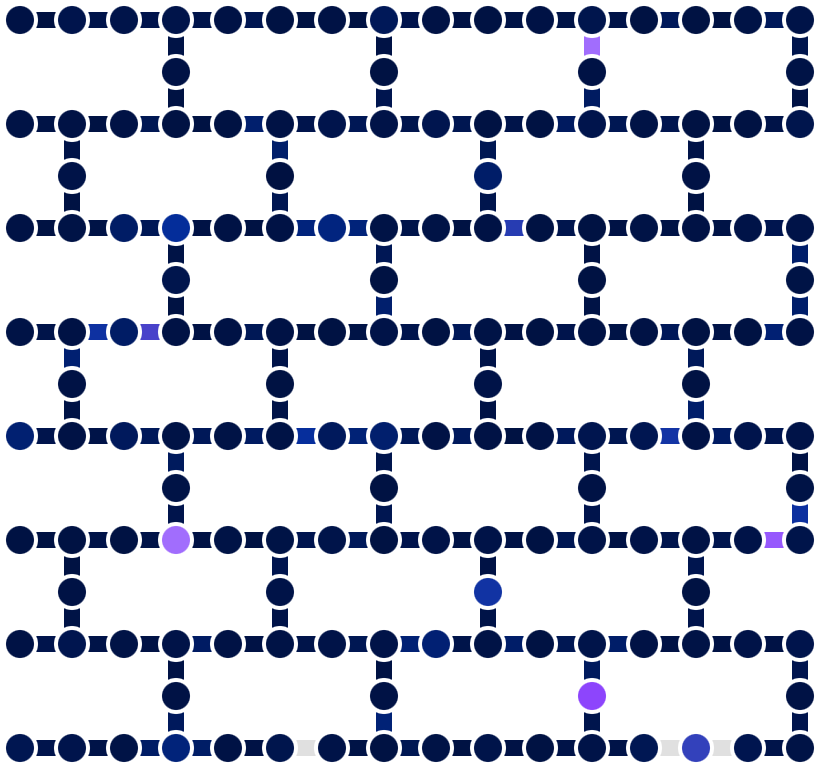}  \caption{The layout of \textit{ibm\_quebec}} 

\label{fig:qpu_layout} 
\end{figure}

\end{document}